\documentclass[aps,pre,twocolumn,groupedaddress]{revtex4-1}

\usepackage{tikz}
\usepackage{graphicx}
\usepackage{ifthen}
\usepackage{amsthm,amsmath,amsfonts}
\usepackage{comment}
\usetikzlibrary{calc}
\usepackage{todonotes}
\usepackage[capitalise]{cleveref}

\usepackage[utf8]{inputenc}
\usepackage[T1]{fontenc}

\newtheorem{definition}{Definition}
\newtheorem{theorem}{Theorem}
\newtheorem{lemma}{Lemma}

\newtheorem{proposition}{Proposition}

\newcommand{\calO}{\mathcal{O}}
\newcommand{\n}{n}
\newcommand\fl[1]{\lfloor #1\rfloor}
\newcommand\E{\mathbb{E}}
\newcommand\Var{\operatorname{Var}}

\newcommand{\new}[1]{#1}

\begin{document}

\title{Social Balance on Networks:~Local Minima and Best Edge Dynamics}

\author{Krishnendu Chatterjee}
\affiliation{IST Austria, Am Campus 1, 3400 Klosterneuburg, Austria}

\author{Andreas Pavlogiannis}
\affiliation{Department of Computer Science, Aarhus University, Aabogade 34, 8200 Aarhus, Denmark}

\author{Jakub Svoboda}
\affiliation{IST Austria, Am Campus 1, 3400 Klosterneuburg, Austria}

\author{Josef Tkadlec}
\affiliation{Department of Mathematics, Harvard University, Cambridge, MA 02138, USA}

\author{\DJ or\dj e \v{Z}ikeli\'c}
\affiliation{IST Austria, Am Campus 1, 3400 Klosterneuburg, Austria}

\begin{abstract}
Structural balance theory is an established framework for studying social relationships of friendship and enmity.
These relationships are modeled by a signed network whose energy potential measures the level of imbalance,
while stochastic dynamics drives  the network towards a state of minimum energy that captures social balance.
It is known that this energy landscape has local minima that can trap socially-aware dynamics, preventing it from reaching balance.
Here we first study the robustness and attractor properties of these local minima.
We show that a stochastic process can reach them from an abundance of initial states,
and that some local minima cannot be escaped by mild perturbations of the network.
Motivated by these anomalies, we introduce Best Edge Dynamics (BED), a new plausible stochastic process.
We prove that BED always reaches balance, and that it does so fast in various interesting settings.
\end{abstract}

\maketitle

\section{Introduction}

The formation of social relationships is a complex process that has long fascinated researchers.
It is well-understood that, besides pairwise interactions, friendships and rivalries are affected by social context.
The study of such phenomena dates back to Heider's theory of \emph{social balance}~\cite{Heider1944,Heider1946,Heider1958}, which can be seen as a rigorous realization of the proverb ``the enemy of my enemy is my friend''.
The theory classifies a social state as \emph{balanced} whenever every group of three entities (a \emph{triad}) is balanced:~it consists of either three mutual friendships, or one friendship whose both parties have a mutual enemy.
The other types of triads create social unrest that eventually gets resolved by changing the relationship between two parties.
For example, a triad with three mutual enmities will eventually lead to two entities forming an alliance against the common enemy.
A triad with exactly one enmity will either see a reconciliation to a friendship under the uniting influence of the common friend, or
lead to the break of one friendship following the social axiom ``the friend of my enemy is my enemy''~\cite{Schwartz2010}.

Cartwright and Harary developed a graph-theoretic model of Heider's theory~\cite{Cartwright1956},
and showed that any balanced state is either a utopia without any enmities,
or it consist of two mutually antagonistic groups~\cite{Harary1953,Davis1967}.
\new{This structural theory of social balance has seen applications across various fields ranging from
philosophy, sociology, or political science~\cite{bramson2017understanding,Taylor1970,fontan2021signed,Sherwin1971,Moore1979}
all the way to fields such as neuroscience or computer science~\cite{Chiang2020,MinhPham2020,li2013influence,altafini2012consensus},
It has also been supported by empirical evidence~\cite{Rawlings2017,Kirkley2019,askarisichani20201995}, see also~\cite{zheng2015social} for a review.
The setting is attractive to physicists due to its intimate connection to the Ising model and spin glasses~\cite{Facchetti2011},
and indeed tools and techniques from statistical physics have proved to be instrumental in improving our understanding of such systems~\cite{belaza2017statistical,saeedian2017epidemic,manshour2021dynamics}, see also~\cite{castellano2009statistical} for a review.}

It is natural to associate each network state with a \emph{potential energy} 
that counts the difference of imbalanced minus balanced triads;
hence the perfectly balanced states are those that minimize the energy of the network~\cite{Marvel2009}.
\new{Understanding how energy is minimized in a system is a fundamental problem studied across different physics fields,
and signed graphs present a clean theoretical framework to study this problem in a setting with a population structure.
It is well known that the energy landscape over signed graphs} has local minima (also known as \emph{jammed states})~\cite{Antal2005}, that is,
states from which all paths to social balance must temporarily increase the number of imbalanced triads.

When the network state is imbalanced, we expect that a social process will perturb it until balance is reached.
The seminal work~\cite{Antal2006} introduced a stochastic process known as \emph{Local Triad Dynamics (LTD)},
according to which imbalanced triads are sampled at random, and the sampled triad is balanced by flipping the relationship of two of its entities. This step is called an \emph{edge flip}.
The same work also introduced \emph{Constrained Triad Dynamics (CTD)}, a socially-aware variant of LTD under which an edge flip is only possible if it reduces the number of imbalanced triads.
Unfortunately, the existence of local minima in the energy landscape implies that CTD can get stuck in jammed states and thus remain permanently imbalanced.

Although the existence of jammed states is well understood in terms of the energy landscape,
little is known about them from the perspective of the stochastic process, that is, about their reachability properties.
For example, from which initial states is it possible to reach a jammed state?
Moreover, if a jammed state is reached, can the process escape if we slightly perturb the network?
Finally, is there a plausible, socially-aware stochastic dynamics (like CTD) that always reaches balance (unlike CTD)?
We tackle these questions in this work.

First, we study the robustness and attractor properties of the local minima of the energy landscape.
We show that the number of jammed states is super-exponential, compared to the previously known exponential lower-bound,
and that jammed states are reachable from any initial state that is not too friendship-dense.
Moreover, we show that some of those jammed states are strongly attracting: 
even when perturbing
 a constant portion of
 edges adjacent to each vertex, the same jammed state is subsequently reached with probability $1$.
\new{As a byproduct, our results resolve an open problem from~\cite{Marvel2009}.}

Second, we propose a new plausible dynamics called \emph{Best Edge Dynamics (BED)}.
Like CTD, BED is a stochastic process in which edge flips are socially-aware, in the sense that they maximize the number of newly balanced triads (see below for details).
We prove that, unlike CTD, BED always reaches a balanced state from any initial state.
Moreover, we show that BED converges faster to a balanced state than CTD in various interesting settings, such as when started from a state that is already close to being balanced.

\new{Finally, we complement our analytical results with computer simulations in the cases when the initial friendship edges form a random Erd\H{o}s-R\'{e}nyi  network or a random scale-free networks. }

\section{Triad Dynamics in Social Networks}

Balance on social networks is studied in terms of signed graphs.
A {\em signed graph} $G=(V,E,s)$ consists of a finite complete graph $(V,E)$ on $|V|=n$ vertices together with an edge labeling
\[s\colon E\rightarrow \{-1,+1\}.\]
The labeling $s$ assigns to each edge one of the two signs; the edges labeled by $+1$ are {\em friendships} and those labeled $-1$ are {\em enmities}.
\new{Thus each pair of individuals (modeled by vertices) has a defined relationship: either they are friends, or they are enemies.}

Given a signed graph $G=(V,E,s)$, a {\em triad} is a subgraph of $G$ defined by any three of its vertices.
A triad is of {\em type $\Delta_k$} for $k=0,1,2,3$ if it contains exactly $k$ edges labeled $-1$.
A triad is {\em balanced} if its type is $\Delta_0$ or $\Delta_2$. 
Intuitively, a triad is balanced if it satisfies the known proverb ``the enemy of my enemy is my friend''.
For an edge $e$ in $G$, its {\em rank} $r_e$ is the number of imbalanced triads containing $e$.
Finally, a signed graph $G$ is {\em balanced} if each triad in $G$ is balanced, see~\cref{fig:process}.

\begin{figure}[h]
\includegraphics[width=\linewidth]{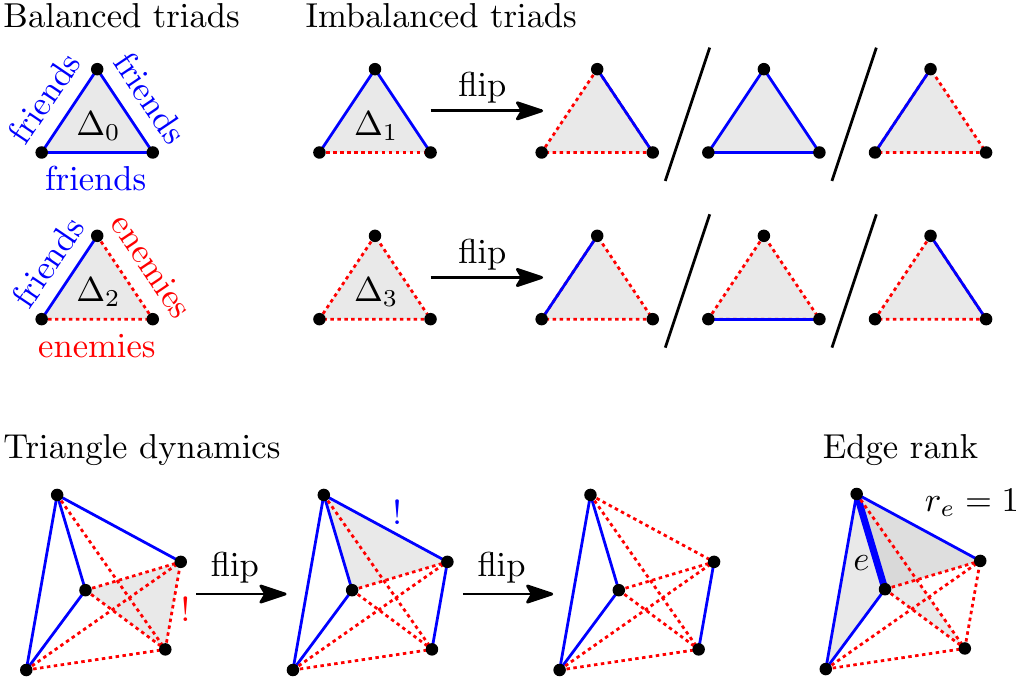}
\caption{A triad of type $\Delta_k$ contains $k$ enmity edges.
The imbalanced triads $\Delta_1$ and $\Delta_3$ can be made balanced by flipping any one edge.
A sequence of flips typically reaches a state where all triads are balanced.
The rank $r_e$ of edge $e$ is the number of imbalanced triads containing $e$.}
\label{fig:process}
\end{figure}

It is known that a signed graph is balanced if and only if its enmity edges form a complete bipartite graph over the vertex set~\cite{Cartwright1956}.
This means that we may partition the vertices of a balanced signed graph into two vertex classes, such that all pairs of vertices from the same class form friendship edges, and all pairs of vertices from different classes form enmity edges.
Moreover, every signed graph which admits such a partitioning is clearly balanced. 
In the special case where one of the vertex classes in this partitioning is empty, each pair of vertices forms a friendship edge and we refer to this balanced signed graph as {\em utopia}.

The main interest in the study of social networks modeled by signed graphs is the evolution of the network
according to some pre-specified dynamics, and the time until the balance is reached. 
The goal is to understand which simple dynamics ensure fast convergence to balance.
Following the work of~\cite{Antal2006}, we focus on those dynamics that, at each time step, select one edge $e$ according to some rule and then flip its sign. 
We then say that ``$e$ is flipped''.
In~\cite{Antal2006}, two such dynamics on signed graphs were introduced:
{\em Local Triad Dynamics (LTD)}, and {\em Constrained Triad Dynamics (CTD)}. 
In the rest of this section, we define these two dynamics and discuss their advantages and limitations.

\subsection{Local Triad Dynamics}

Let $G$ be a signed graph modeling a social network with friendships and enmities. 

The {\em Local Triad Dynamics (LTD)} with parameter $p\in [0,1]$ is a discrete-time random process that starts in $G$ and repeats the following procedure until there are no imbalanced triads in $G$:
\begin{enumerate}
\item Select an imbalanced triad $\Delta$ uniformly at random.
\item If $\Delta$ is of type $\Delta_3$, then an edge of $T$ is chosen to be flipped uniformly at random.
If $\Delta$ is of type $\Delta_1$,
then the unique edge with sign $-1$ is chosen to be flipped with probability $p$ and each of the two other edges
is chosen with probability $\frac{1-p}{2}$.
\end{enumerate}
We refer to distinct signed graphs as {\em states}.

LTD is socially oblivious in the sense that once an imbalanced triad is selected, the edge to be flipped is chosen according to a (stochastic) rule that disregards the rest of the network. 
Moreover, the guarantees of LTD on the expected time to reach a balanced state are not very plausible:
it was shown in~\cite{Antal2006} that if $p < \frac{1}{2}$, then the expected time grows exponentially with the size of the signed graph.
On the other hand, if $p> \frac{1}{2}$, then the dynamics is more likely to create rather than remove friendships and it reaches utopia with high probability. 
This means that the eventual balanced state is essentially pre-determined.

\subsection{Constrained Triad Dynamics}

The {\em Constrained Triad Dynamics (CTD)} is another dynamics on signed graphs.
Given a signed graph $G$ on $n$ vertices, CTD is a random process
that starts in $G$ and repeats the following procedure until there are no imbalanced triads in $G$:
\begin{enumerate}
\item Select an imbalanced triad $\Delta$ uniformly at random.
\item Select an edge $e$ of $\Delta$ uniformly at random.
\item Flip $e$ if $r_e \ge \frac{1}{2}n-1$, that is,~if the number of imbalanced triads in $G$ does not increase upon the flip (in the case of equality, the flip happens with probability $1/2$), otherwise do nothing.
\end{enumerate}
Note that CTD introduces a non-local, socially-aware rule: When deciding whether a selected edge should be flipped, we take into account all triads that contain it.
In~\cite{Antal2006} it was claimed that, starting from any initial signed graph $G$,
CTD converges to a balanced state and that balance is reached fast
-- in a time that scales logarithmically with the size $n$ of the signed graph.
If true, this would imply that CTD overcomes the limitation of LTD in which the expected convergence time could be exponential in $n$. 
However the claim, which was supported by an informal argument,
is not quite true, since the energy landscape is rugged: There are states, called {\em jammed states}, that are not balanced but where CTD can not make a move, since any flip would (temporarily) increase the number of imbalanced triads~\cite{Antal2005}.
Moreover, it is known that there are at least roughly $3^n$ jammed states (compared to roughly $2^n$ balanced states)
and that some of the jammed states have zero energy~\cite{Marvel2009}.

\section{Reaching and Escaping the Jammed States}\label{sec:ctd}

Even though there are exponentially many jammed states,
computer simulations on small populations were used to suggest that they can effectively be ignored~\cite{Antal2005}.
In contrast, in this section we present three results which indicate that for large population sizes the jammed states are important.

First (``counting''), we study the number of jammed states.
It is known~\cite{Antal2005}, that there are at least $3^n$ jammed states, that is, at least exponentially many. 
Here we construct a family of simple, previously unreported jammed states and we show that
the total number of jammed states on $n$ labeled vertices is super-exponential, namely at least $2^{\Omega(n\log n)}$.
This shows that even on the logarithmic scale,
the jammed states are substantially more numerous than the balanced states (of which there are ``only''~$2^{n-1}$).

Our second result (``reaching'') shows that, starting from any initial signed graph that is not too friendship-dense,
a specific jammed state $J$, which we construct below, is reached with positive probability.
In particular this implies that the expected time to balance in this stochastic process is formally infinite,
even for signed graphs that have a constant positive \new{density of friendship edges}.

Our third result (``escaping'') shows that this specific jammed state $J$ forms a deep well in the energy landscape: 
Once it is reached, it can not be escaped even if we perturb a constant portion of edges incident to each vertex.

In the rest of this section, we sketch the intuition behind these results.
For the formal statements and proofs, see~\cref{thm:counting,thm:reaching,thm:escaping} in~\cref{app:ctd}, respectively.

Regarding the first result (``counting''), the new jammed states are defined in terms of an integer parameter $d$.
We partition the population into $4d+2$ clusters ($4d+1$ or $4d+3$ would work too), arrange the clusters along a circle,
and assign a sign $+1$ to those (and only those) edges that connect individuals who live in clusters that are at most $d$ steps apart, see~\cref{fig:counting}.
We then show that, for each friendship or enmity edge $(u,v)$, a strict majority $2(d+1)>\frac12(4d+3)$ of clusters have the property
that any vertex $w$ in that cluster forms a balanced triad $(u,v,w)$ with $(u,v)$.
Thus, flipping $(u,v)$ would increase the number of imbalanced triads.
When all the clusters are roughly equal in size, which is the typical behavior for large population sizes, the state is thus jammed.
Our construction also resolves in affirmative an open question~\cite{Marvel2009} which asks whether there exist jammed states with an even number of friendship cliques (here this number is $4d+2$).

\begin{figure}[h]
\includegraphics[width=\linewidth]{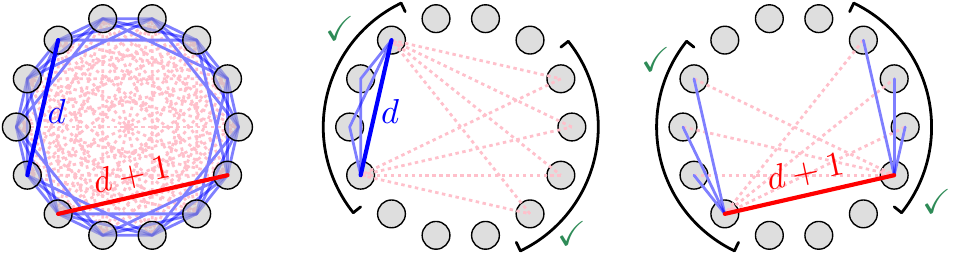}
\caption{A jammed state consisting of $4d+2$ roughly equal clusters, each connected by friendships to the clusters at most $d$ steps apart.}
\label{fig:counting}
\end{figure}

Regarding the second result (``reaching''),
we consider any initial state $I_n$ on $n$ vertices in which each vertex is incident to at most $\n/12-1$ edges labeled $+1$.
We define a jammed state $J_n$ as follows: We partition the vertices into three clusters $V_1$, $V_2$, $V_3$ of roughly equal size,
label all edges within each set $+1$ and all other edges $-1$.
Then we exhibit a sequence of flips that transforms $I_n$ into $J_n$. This is done in two phases:
First, one can verify that any time we select an imbalanced triangle that contains an edge labeled $-1$ within one cluster $V_i$, this edge can be flipped.
Hence, we may flip all enmity edges within the three clusters to reach a state in which all edges within each $V_i$ are labeled $+1$.
After that, one can similarly verify that all edges labeled $+1$ that connect vertices in two different parts $V_i$, $V_j$ can be flipped one by one, thereby reaching the jammed state~$J_n$, see~\cref{fig:reaching}.

\begin{figure}[h]
\includegraphics[width=\linewidth]{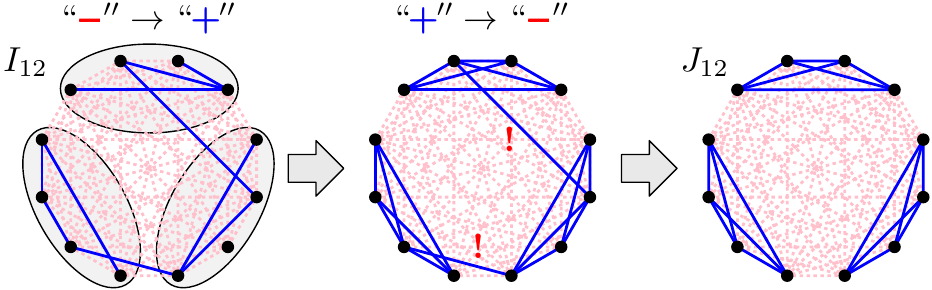}
\caption{A jammed state $J_n$ (right) on $n$ vertices can be reached from any not too friendship-dense initial state $I_n$ (left). Moreover, once it is reached, it can not be escaped, even if a substantial portion of edges around each vertex are perturbed.}
\label{fig:reaching}
\end{figure}

Regarding the third result (``escaping''),
we consider any state $S_n$ that can be obtained from $J_n$ by flipping a set $E_0$ of edges such that each vertex is incident to at most $\n/12-1$ edges of $E_0$.
We then show that edges that do \textit{not} belong to $E_0$ can never be flipped.
On the other hand, each edge that does belong to $E_0$ can be flipped.
Thus all edges in $E_0$ are eventually flipped back and the jammed state $J_n$ is reached again.

\section{Best Edge Dynamics}

Our results in the previous section show that the jammed states are a profound feature of the energy landscape:
They are reachable from many conceivable initial states,
and some of them trap stochastic dynamics such as CTD forever even if we allow substantial perturbations.

This brings up a question of whether there exists a simple socially-aware dynamics with all the desirable properties of CTD, but one that can not get stuck in a jammed state. 
To address this question, we propose the {\em Best Edge Dynamics (BED)}, a modification of CTD that, unlike CTD, reaches a balanced state with probability $1$ from every initial state. 
Moreover, we prove that BED converges fast to a balanced state in several important cases (see~\cref{prop:red-fast,prop:jammed-fast}),
and our empirical evaluation of both BED and CTD in Section~\ref{simulation_results} shows that in general the convergence times are comparable (after we exclude the runs where CTD does not terminate).

Let $G$ be a signed graph. Then the {\em Best Edge Dynamics (BED)} is a discrete-time random process
that starts in $G$ and repeats the following procedure until there are no imbalanced triads in $G$:
\begin{enumerate}
\item Select an imbalanced triad $\Delta$ uniformly at random.
\item Select an edge $e$ from $\Delta$ with the highest rank $r_e$ (in case of ties, pick one such edge uniformly at random).
\item Flip $e$.
\end{enumerate}

Note that, in contrast to CTD, we flip $e$ even when its rank $r_e$ satisfies $r_e<\frac12|V|-1$, that is, when flipping the best edge creates more imbalanced triads than it removes.
In particular, whenever we reach a jammed state, we still make a flip.
In principle, it could still happen that BED remains trapped in a subset of imbalanced states,
toggling edges back and forth unable to escape it,
but in fact we prove that this event occurs with probability 0.

\begin{theorem}\label{thm:bedprob1}
For any initial signed graph on $n$ vertices, BED reaches a balanced state with probability $1$ and in finite expected time.
\end{theorem}

To prove~\cref{thm:bedprob1}, we will show that any signed graph on $n$ vertices can become balanced upon $\calO(n^3)$ flips.
This suffices since BED induces a finite Markov chain over the set of all states, and the absorbing states of the Markov chain are precisely the balanced signed graphs.

Fix an edge $(v_1,v_2)$ in $G$ of the lowest rank, and set $B=\{v_1,v_2\}$. 
Note that it is possible to flip one edge in each imbalanced triad containing $(v_1,v_2)$ without flipping $(v_1,v_2)$ itself, to make all triads containing $(v_1,v_2)$ balanced:
Indeed, if we consider an imbalanced triad containing $v_1$, $v_2$ and some third vertex $w$, as BED flips an edge of the highest rank, it can flip either $(v_1,w)$ or $(v_2,w)$. This makes the triad balanced and decreases $r_{(v_1,v_2)}$ by $1$, while decreasing the rank of any edge that hasn't been flipped by at most $1$. Hence, $(v_1,v_2)$ will still be of the lowest rank in all imbalanced triads containing it, so we can flip one edge in each such triad until all triads containing $(v_1,v_2)$ become balanced.

The rest of the construction proceeds inductively by adding a new vertex to $B$ in each step and making all triads which contain an edge with endpoints in $B$ balanced. 
The process ends when all vertices of $G$ have been added to $B$. For the inductive step, suppose that every triad containing at least two vertices in $B$ is balanced.
Let $(v,w)$ be an edge in $G$ which is of lowest rank among all edges with $v\in B$ and $w\not\in B$. 
Each triad containing $w$ and two vertices in $B$ is balanced. 
On the other hand, since $(v,w)$ is an edge of lowest rank among all edges with $v\in B$ and $w\not\in B$,
it follows that for each imbalanced triad $\Delta$ containing $v$, $w$ and a third vertex $u$, 
BED can flip either $(w,u)$ or each edge $(v',u)$ with $v'\in B$ to make $\Delta$ and all other triads containing $u$ and two vertices in $B$ balanced. 
By doing this for each imbalanced triad containing $(v,w)$, we modify the signs in the graph in such a way that all triads containing at least two vertices in $B\cup \{w\}$ become balanced. Thus we can add $w$ to $B$. By induction on the size of $B$, this way we eventually reach a balanced state.

Notice that in each iteration of the above construction, at most $n\cdot |B|$ edges are flipped. Hence the total number of edge flips is at most $n\sum_{i=1}^ni=\calO(n^3)$ as claimed.

\subsection{Fast convergence and red-black graphs}

So far we have shown that, unlike CTD, BED ensures convergence to balance with probability $1$ and in finite expected time.
In the rest of this section we show that BED also provides theoretical guarantees on fast convergence when started in certain states that are either ``close'' to being balanced (\cref{prop:red-fast}) or jammed (\cref{prop:jammed-fast}), showing that this new dynamics is robust.

We start by introducing the {\em red-black graphs},
a new concept that allows neat reasoning about signed graphs that are close to being balanced.
Given a signed graph $G$, let $C$ be a balanced signed graph on the same number of vertices which differs from $G$ in the smallest number of edge signs.
We refer to $C$ as a {\em closest} balanced state to~$G$.
Then the red-black graph $R$ associated to $G$ and $C$ is obtained from $G$ by
coloring each edge of $G$ in black if the signs of the edge in $G$ and $C$ agree,
and in red otherwise.
Thus, red edges are precisely those edges whose signs in $G$ and $C$ are misaligned.
Figure~\ref{figure:redBlack} shows an example of a signed graph and the corresponding red-black graph.

\begin{figure}[h]
\includegraphics[width=\linewidth]{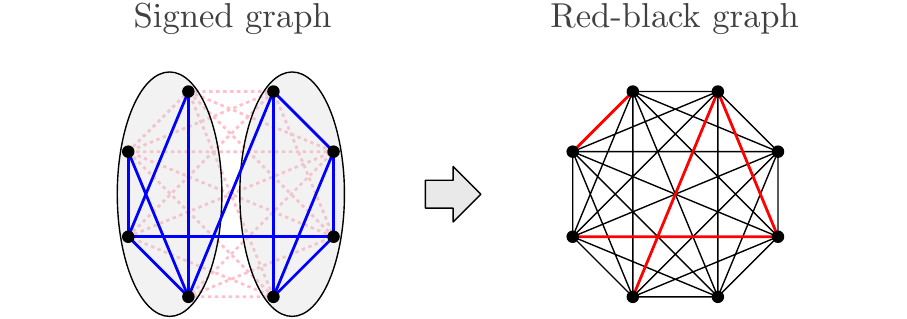}
\caption{An example of a signed graph (left) and the corresponding red-black graph (right).
}
\label{figure:redBlack}
\end{figure}

The key property of red-black graphs is that red and black edges can be viewed as enmities and friendships when reasoning about balanced triads in the following sense: A triad in $G$ is imbalanced if and only if exactly $1$ or $3$ of its edges are red in $R$.
The proof of this claim is by casework and is deferred to Appendix~\ref{app:redblack}, see~\cref{lemma:redblack}.
This also implies that the rank of an edge in $G$ is equal to its rank in $R$ if we treat red edges in $R$ as enmities and black edges as friendships.

\subsection{Fast convergence around balanced states}\label{sec:fast-balanced}

We are now ready to study the convergence of BED when started in a signed graph which is close to being balanced.

\begin{proposition}\label{prop:red-fast}
Consider a signed graph $G$ whose red-black graph $R$ satisfies one of the following two conditions:
\begin{enumerate}
\item Each vertex is incident to at most $\frac14n-1$ red edges.
\item There are at most $\frac12n$ vertices incident to a red edge.
\end{enumerate}
Then BED reaches a balanced state in $\calO(n^2)$ steps in the worst case.
\end{proposition}

To prove the proposition, it suffices to show that BED flips only red edges:
Since in total there are $\calO(n^2)$ red edges, BED reaches a balanced state in $\calO(n^2)$  steps.

Consider the imbalanced triad $(u,v,w)$ selected by BED.
If the triad contains $3$ red edges, clearly BED flips a red edge.
Otherwise, from the key property of red-black graphs (\cref{lemma:redblack}) we know that the triad contains exactly $1$ red edge.
Without loss of generality suppose $e=(u,v)$ is the red edge.
For each of the two conditions in~\cref{prop:red-fast} we argue separately:
\begin{enumerate}
\item Since $(v,w)$ is a black edge, any triad containing $(v,w)$ is imbalanced if it contains precisely $1$ red edge. 
There are at most $\frac{1}{2}n-2$ such triads, since there are at most $\frac{1}{4}n-1$ red edges incident to $v$ and similarly at most $\frac{1}{4}n-1$ to $w$. Thus, $r_{(v,w)}\leq \frac{1}{2}n-2$. 
Analogously $r_{(u,w)}\leq \frac{1}{2}n-2$. On the other hand, as $(u,v)$ is red, a triad containing $(u,v)$ which is balanced has to contain exactly two red edges. 
So $(u,v)$ is contained in at most $\frac{1}{2}n-2$ balanced triads by the same argument as above, and
$r_{(u,v)} \geq n-2 - (\frac{1}{2}n-2)=\frac{1}{2}n$. 
Thus
\[r_{(u,v)} > \max\{r_{(u,w)},r_{(v,w)}\},\]
so BED will flip the red edge.

\item By assumption, we can partition the vertices of $G$ into two sets $V_1$ and $V_2$ such that
$|V_1| \ge |V_2|$ and all red edges have both endpoints in $V_2$.
Then any triad $(u,v,w')$ with $w'\in V_1$ contains exactly~$1$ red edge and is imbalanced, so $r_e\geq |V_1|$.
On the other hand, as $(u,w)$ is black, for any third vertex contained in $V_1$ the triad is balanced,
so $r_{(u,w)}\leq |V|-2-(|V_1|-1)=|V_2|-1<|V_1|$.
Analogously $r_{(v,w)}<|V_1|$, so $e$ has the highest rank in $(u,v,w)$ and BED flips $e$.
\end{enumerate}

\subsection{Fast convergence from jammed states}

Recall that a state is jammed if it is not balanced but CTD cannot flip an edge in any imbalanced triad.
Here we show that from certain jammed states, BED converges to a balanced state after $\calO(n^2)$ edge flips, in expectation.
Thus, BED ensures fast convergence even when the convergence time of CTD is infinite.
(For details, see Appendix~\ref{app:fastbedstuck},~\cref{prop:jammed-fast}.)

As before, consider the jammed state $J_n$ consisting of three large roughly equal clusters of friends on $n$ vertices in total.
\cref{fig:Jn-fast} illustrates that, started from $J_n$, BED converges to balance in $\calO(n^2)$ time.
Initially, BED keeps adding friendship edges connecting different clusters.
Due to random fluctuations, the symmetry among the three clusters breaks and
one pair of clusters becomes more densely connected than the other pairs.
This difference is exaggerated over time and eventually that pair of clusters merges.

\begin{figure}[h]
\includegraphics[width=0.5\linewidth]{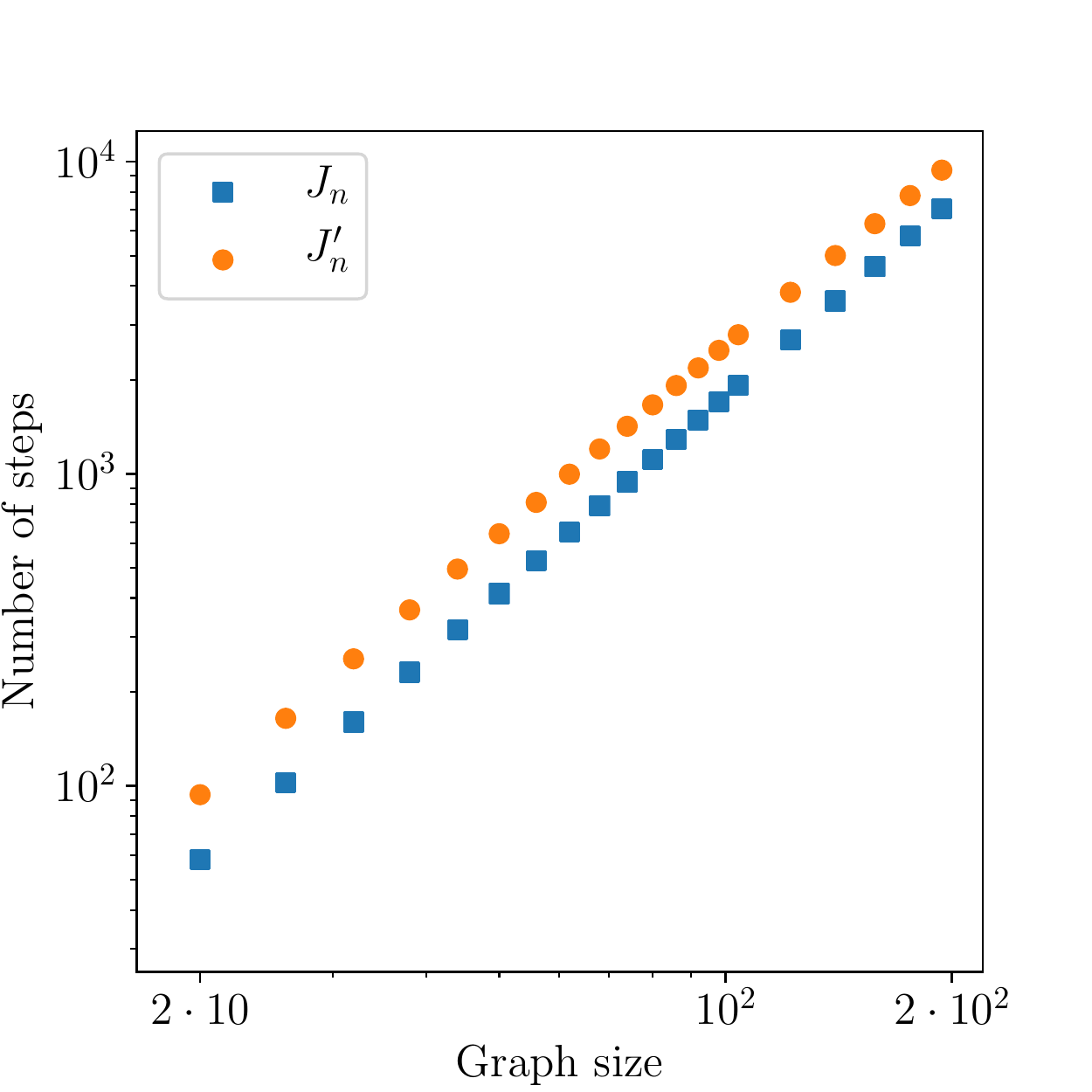}%
\includegraphics[width=0.5\linewidth]{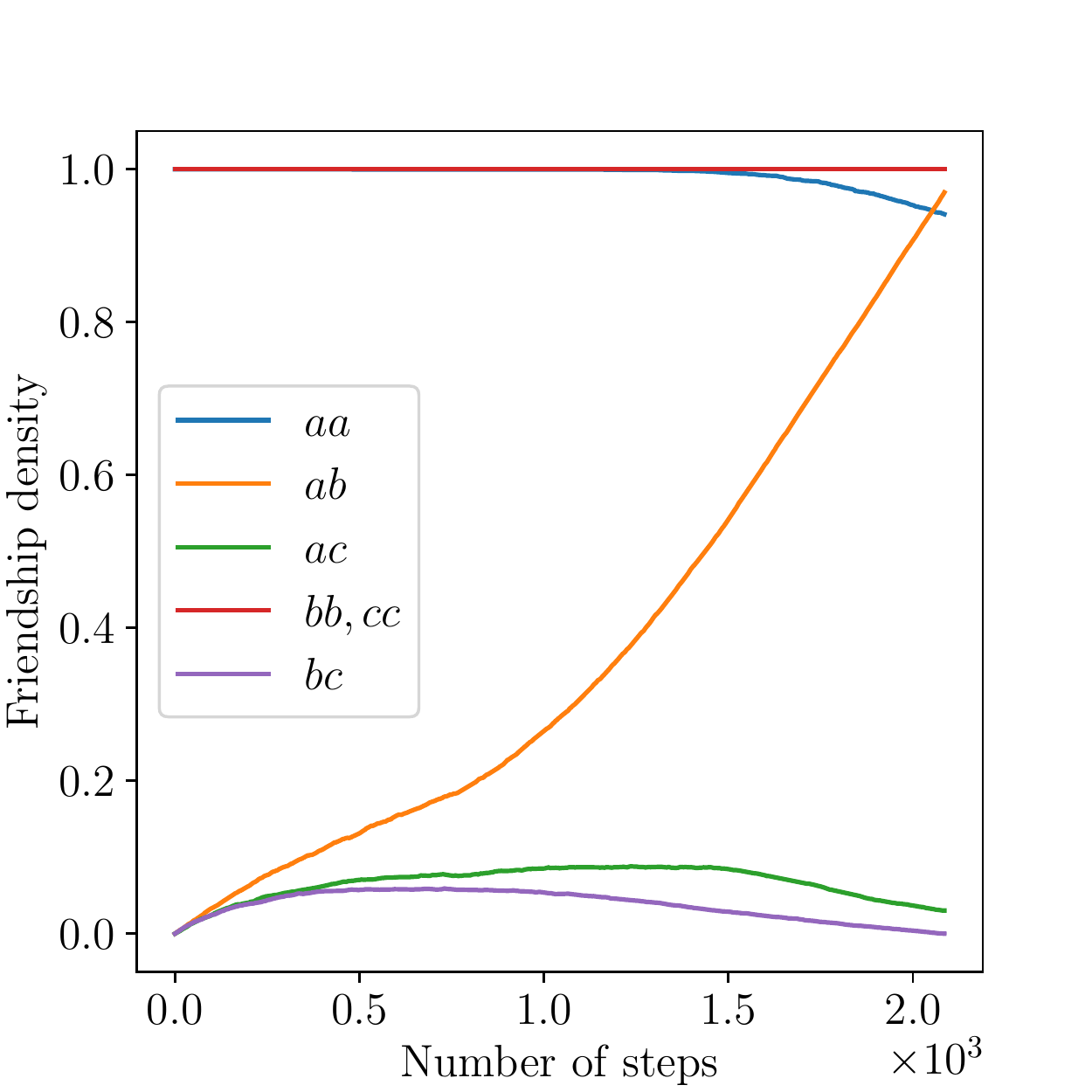}
\caption{
Left: Started from $J_n$, BED reaches balance in $\calO(n^2)$ expected steps.
Right: Friendship densities in different portions of the signed graph $J_{100}$, in a single run of BED.
Apart from possibly the very end, the friendships within clusters ($aa$, $bb$, $cc$) are never flipped.
Eventually, one pair of clusters (here $a$ and $b$) merges.
}
\label{fig:Jn-fast}
\end{figure}

Next, we consider a state $J'_n$ whose $n$ vertices are split into 6 clusters
 arranged along a circle with relative sizes roughly $2:1:1:2:1:1$.
Two vertices are connected by a friendship edge if they belong to the same cluster or to adjacent clusters (see~\cref{fig:bed-jammed}, left).
The different cluster sizes ensure that the symmetry is broken from the very beginning and allow for a simpler formal argument.

\begin{figure}[h]
\includegraphics[width=\linewidth]{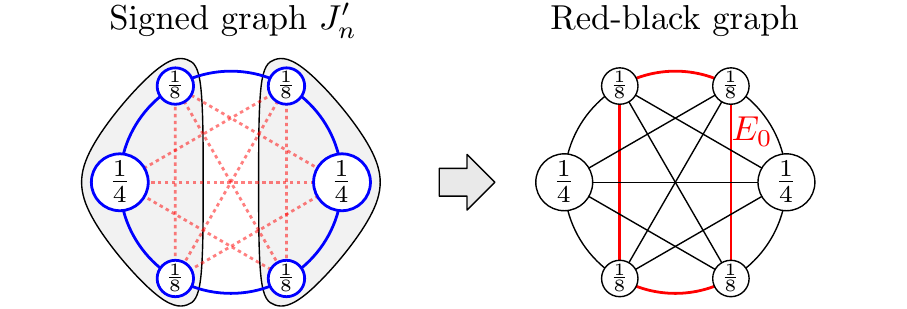}
\caption{A jammed state $J'_n$ from which BED reaches a balanced state in $\calO(n^2)$ expected steps.
}
\label{fig:bed-jammed}
\end{figure}

It is straightforward to check that the state $J'_n$ is jammed and that the closest balanced state $C$ is
the one depicted in the left figure. The corresponding red-black graph is shown on the right.
Let $E_0$ be the set of edges that are initially red in the red-black graph.
We show that the process always flips an edge $e\in E_0$
and that, at each point in time, we are a constant-factor more likely to turn a red edge into a black one rather than the other way around.
Thus the stochastic process can be projected onto a random walk with a constant forward bias.
Since such a random walk terminates in the number of steps that is linear in its length,
this proves that the process finishes in $\calO(|E_0|)=\calO(n^2)$ steps in expectation.

\new{\section{Computer Simulations}\label{simulation_results}}
\new{In this section, we compare the two dynamics CTD and BED by means of computer simulations.
In each simulation, we generate a network (possibly randomly) and assign ``$+$'' to each its edge.
All other edges are assigned ``$-$'', so the underlying network is always a complete graph.
Then we simulate each dynamics to determine the quantities such as the typical outcome and the number of steps until it reaches balance.}

\smallskip
\subsection{Erd\H{o}s-R\'{e}nyi graphs}
\def\ER{\operatorname{ER}}
First we consider random Erd\H{o}s-R\'{e}nyi graphs $\ER(n,p)$, where each two of the $n$ vertices are connected by a friendship edge with probability~$p$, independently of each other.

Since CTD can get jammed, the average number of steps until it reaches balance is infinite for many initial states.
To have a meaningful comparison, we first exclude runs in which CTD gets jammed (later we report their proportion).
Upon this exclusion, the two dynamics are comparable.
First, the number of steps until balance for both dynamics scales as $\Theta(n^2)$ with the population size $n$, see~\cref{figure:stepsPlot}.

\begin{figure}
  \begin{center}
    \includegraphics[width=0.8\linewidth]{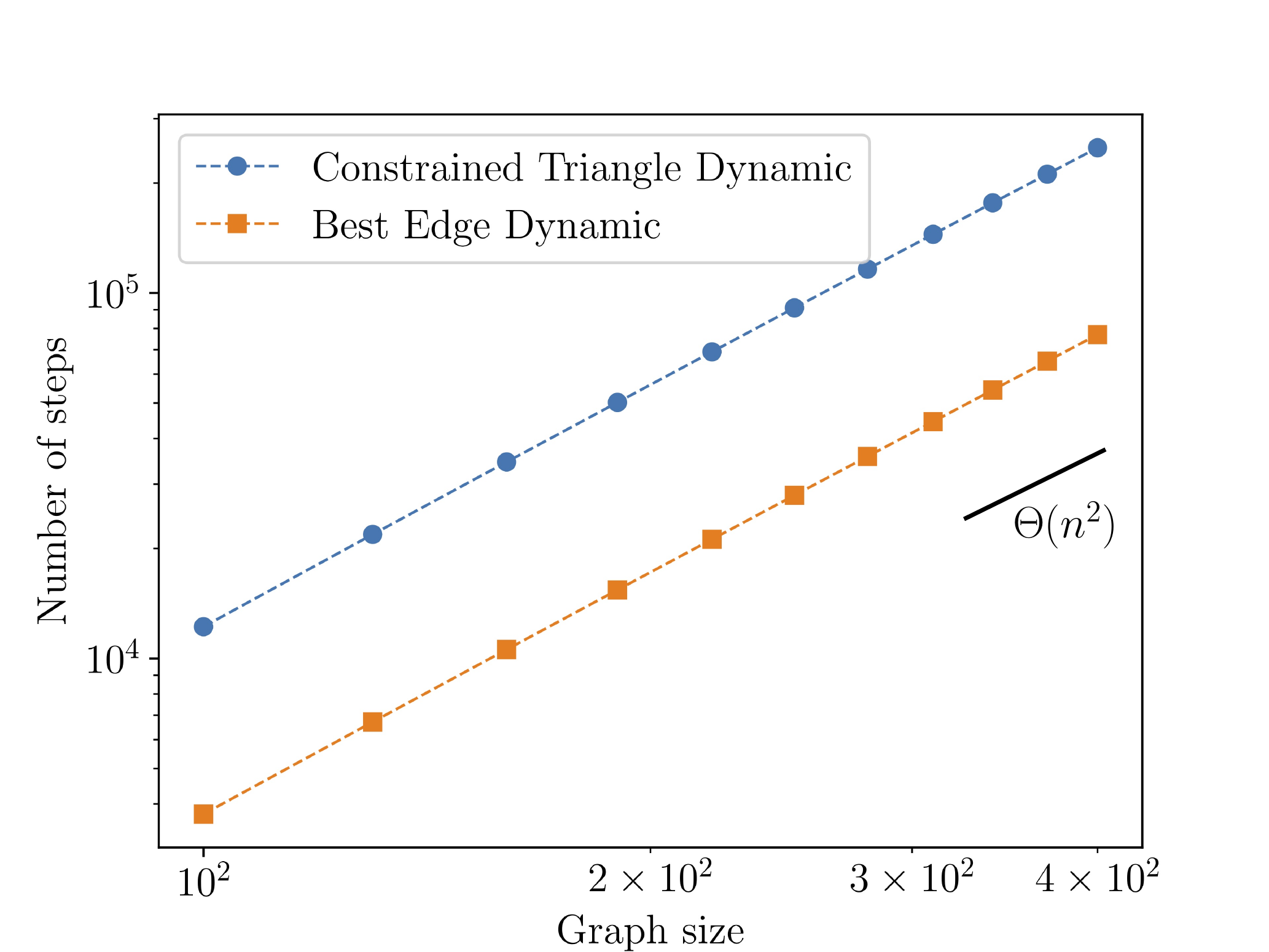}
    \caption{Average number of steps until balance for CTD (excluding the runs that get jammed) and BED, over $10^5$ runs.
    The friendships in the initial signed graphs form the Erd\H{o}s-R\'{e}nyi graph with edge probability $p = \frac12$ and size $n\le 400$.
    Both quantities scale as $\Theta(n^2)$.
    }
    \label{figure:stepsPlot}
  \end{center}
\end{figure}

\new{Second, the final configuration does not depend on the choice of the dynamics
but it is strongly dependent on the parameter $p$.
Namely, for $p\le 0.5$ the two cliques are almost always roughly equal in size,
whereas for $p\ge 0.6$ the larger clique contains almost all the vertices, see~\cref{figure:like_6}.
(See also~\cref{app:table} for tables showing several network descriptors before and after the network becomes balanced.)}

\begin{figure}
  \begin{center}
    \includegraphics[width=0.5\linewidth]{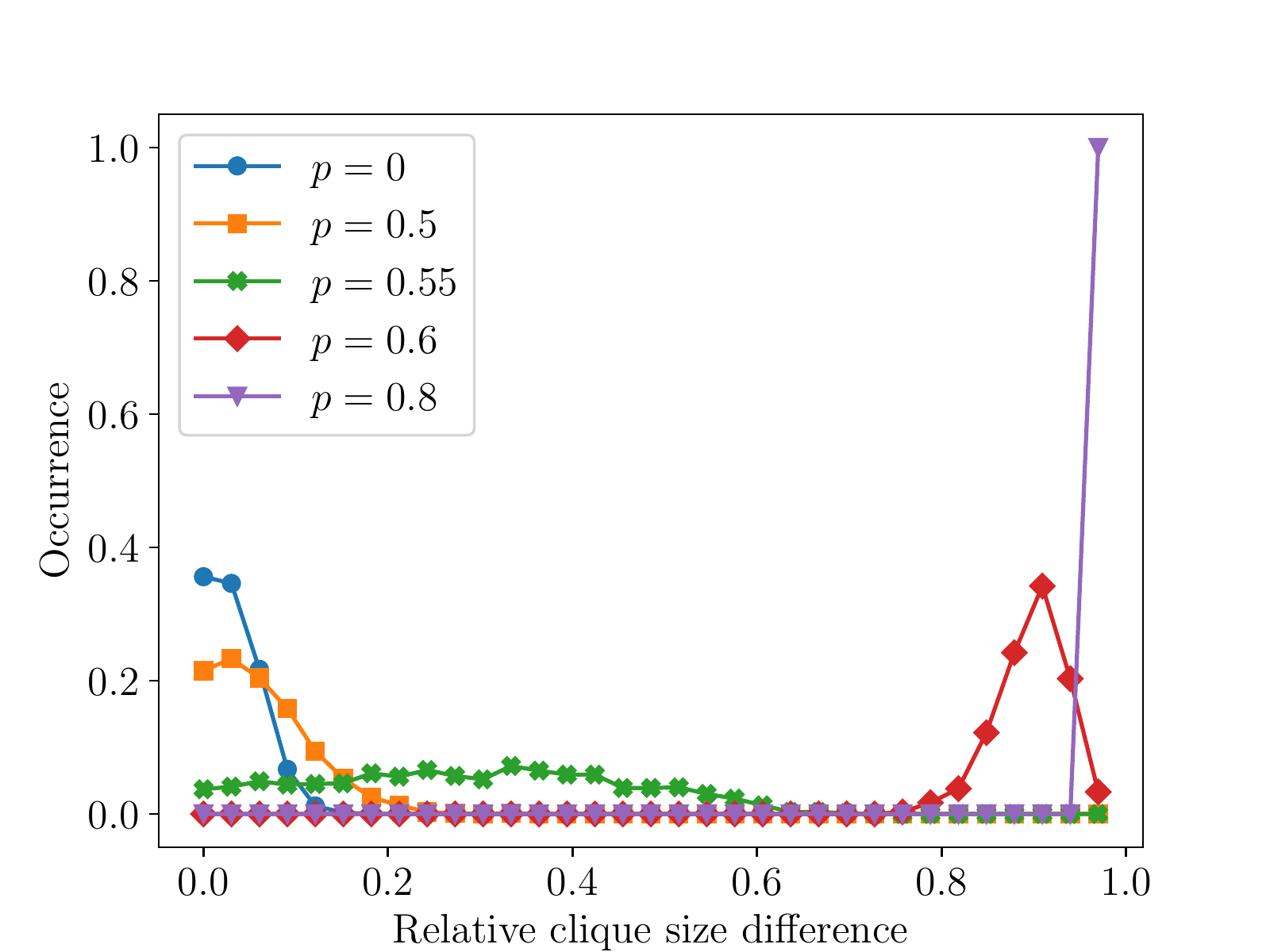}%
    \includegraphics[width=0.5\linewidth]{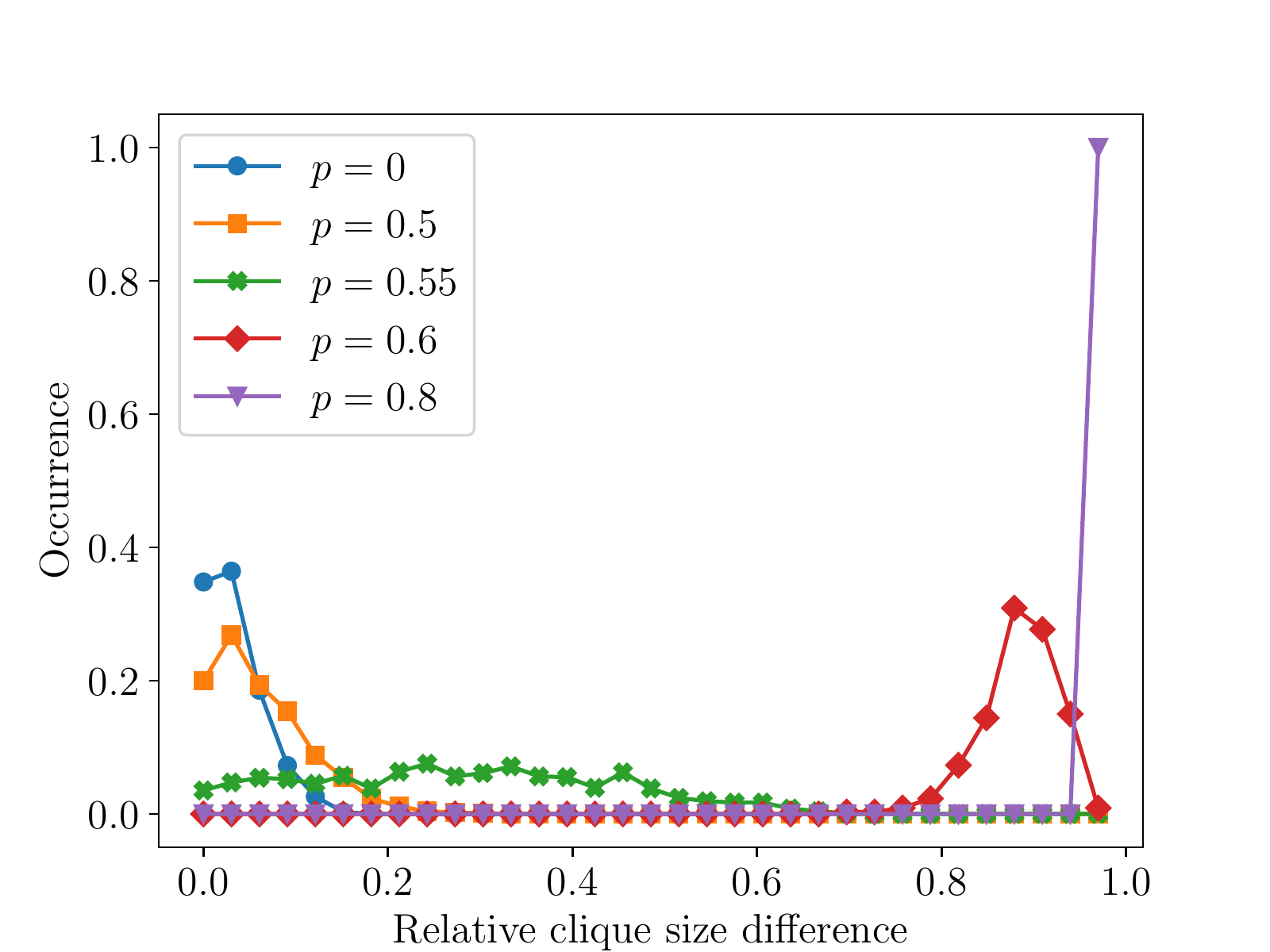}
    \caption{\new{Distribution of the relative difference of the clique sizes once balance is reached,
    for BED (left) and CTD (right) over $10^5$ runs. Here $n=128$.}
    }
    \label{figure:like_6}
  \end{center}
\end{figure}

Next we focus on the probability that CTD gets jammed.
For fixed $n$, this probability exhibits a threshold behavior as a function of the friendship density $p$, see~\cref{figure:vary_p}.
The intuition is as follows:
When the initial friendship density is large (here $p\ge 0.6$), then the initial state is close to utopia (the balanced state that consists of only friendships).
Utopia is then reached quickly and with high probability (cf.~\cref{sec:fast-balanced}).
When the initial friendship density is small (here $p\le 0.5$), the jamming probability is nonzero (cf.~\cref{sec:ctd}).
Most imbalanced triads are of type $\Delta_3$ (all enmities).
The dynamics thus keeps adding friendship edges and the jamming probability is mostly independent of $p$.
The same phenomenon occurs for other sizes $n$, see~\cref{figure:probabilityPlot}.

\begin{figure}
  \begin{center}
    \includegraphics[width=0.8\linewidth]{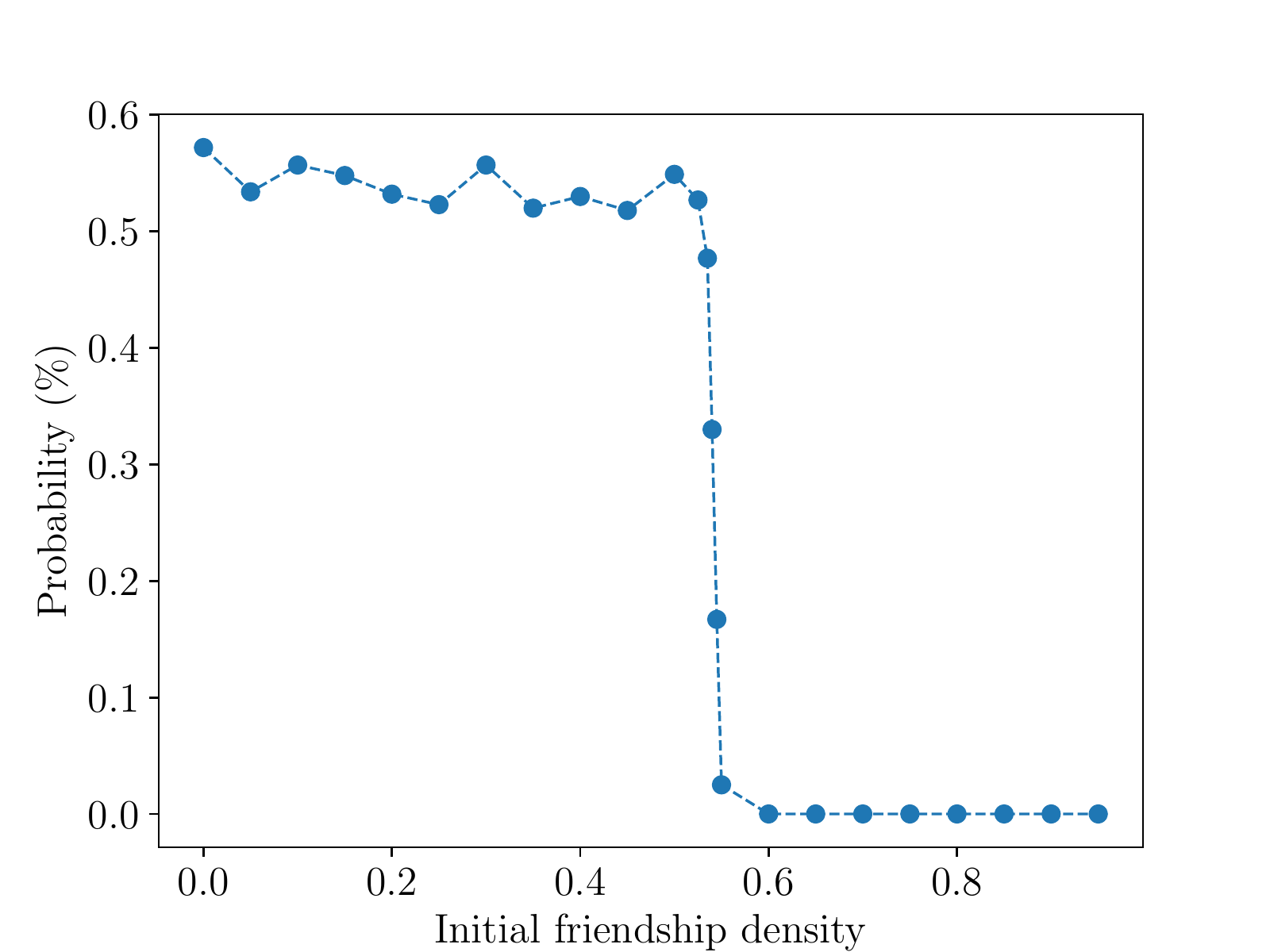}
    \caption{The jamming probability for CTD,
    when friendships form an Erd\H{o}s-R\'{e}nyi graph with size $n=250$ and edge density $p\in[0,1]$
    exhibits a threshold behavior.
}\label{figure:vary_p}
  \end{center}
\end{figure}

\begin{figure}
  \begin{center}
    \includegraphics[width=0.8\linewidth]{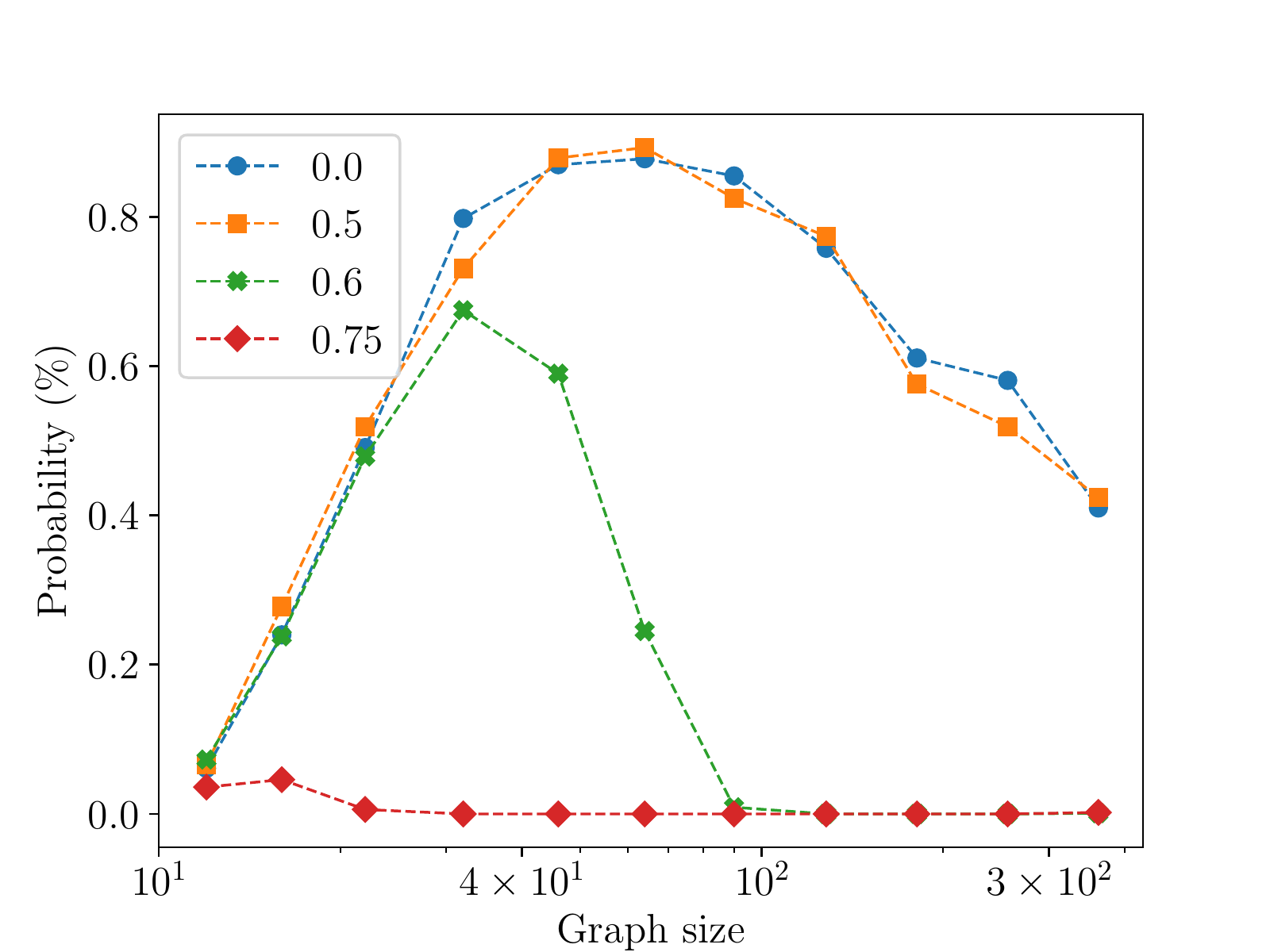}
    \caption{The jamming probability for CTD, when friendships form an Erd\H{o}s-R\'{e}nyi graph with size $n\le 400$ and edge density $p\in\{0,0.5,0.6,0.75\}$.
    When $p\ge 0.75$ (or $p\ge 0.6$ and $n$ large), the dynamics typically reaches utopia, otherwise there is a non-negligible probability of reaching a jammed state.
    }\label{figure:probabilityPlot}
  \end{center}
\end{figure}

\subsection{Scale-free networks}
\def\BA{\operatorname{BA}}
Apart from Erd\H{o}s-R\'{e}nyi graphs we also consider the Barab\'asi-Albert model $\BA(n,d)$ for scale-free networks~\cite{Barabasi1999}.
This model creates a scale-free graph with edge density $d\in[0,1]$.
In particular, we start with a path on $11$ vertices and then process the remaining $n-11$ vertices one by one.
When processing a vertex $v$, we randomly connect it to a subset of vertices already present in the network in such a way that
the probability that a pair $uv$ forms a (friendship) edge is proportional to the degree of $u$ (``preferential attachment'') and the resulting (expected) edge density equals $d$.

Compared to Erd\H{o}s-R\'{e}nyi graphs, the degrees of the vertices are unequally distributed.
Despite this difference, the number of steps until balance for both CTD and BED still scales as $\Theta(n^2)$ with the size $n$, see~\cref{figure:BA_steps}.
Moreover, the jamming probability for CTD still exhibits a similar threshold behavior, see~\cref{figure:BA_probability_density,figure:BA_probability}.

\begin{figure}
  \begin{center}
    \includegraphics[width=0.8\linewidth]{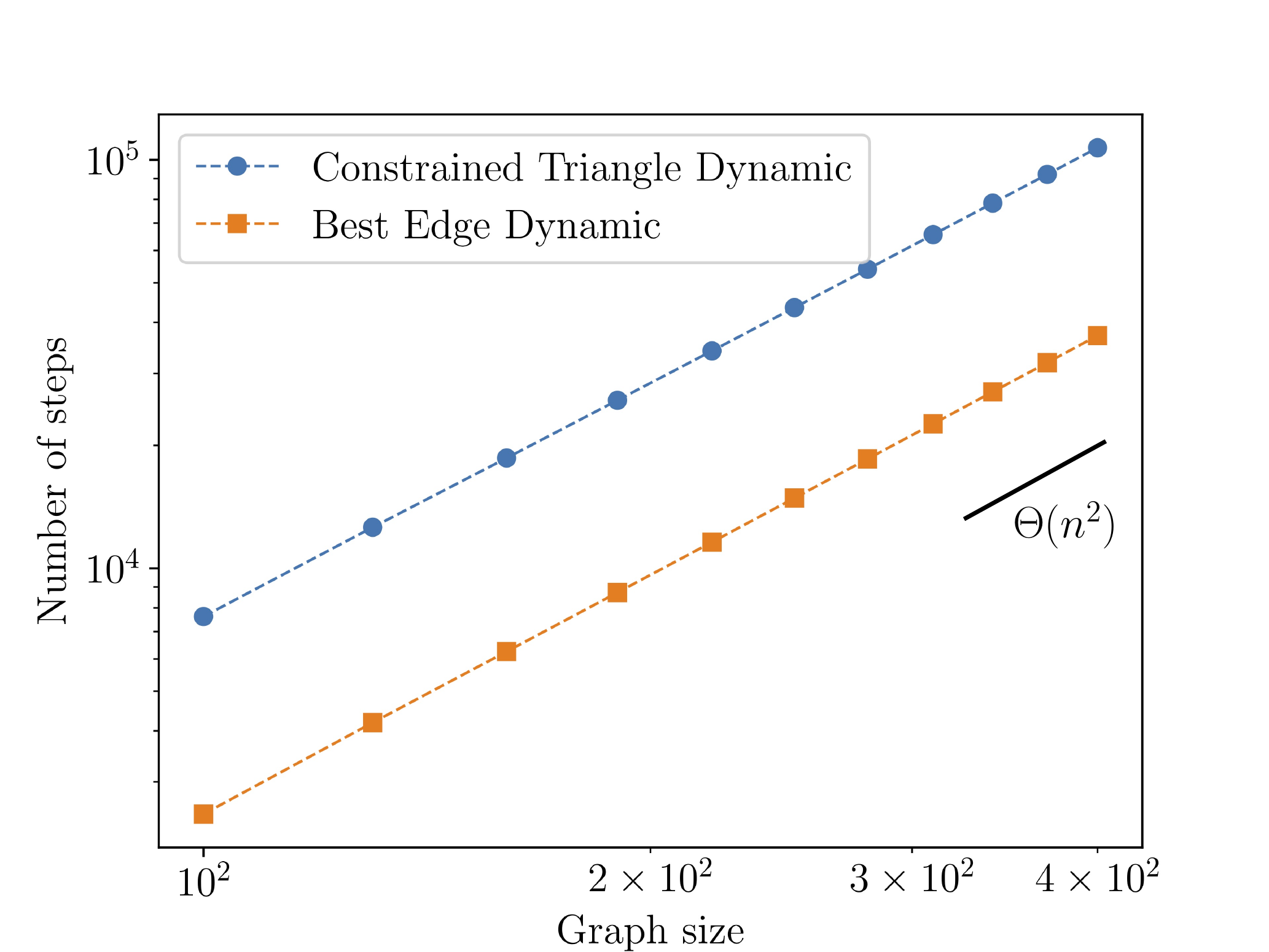}
    \caption{Average number of steps until balance for CTD (excluding the runs that get jammed) and BED, over $10^5$ runs.
    The initial signed graphs are Barab\'asi-Albert with degree parameter $d = 0.5$ and size $n\le 400$.
    Both quantities scale as $\Theta(n^2)$.
    }\label{figure:BA_steps}
  \end{center}
\end{figure}

\begin{figure}
  \begin{center}
    \includegraphics[width=0.8\linewidth]{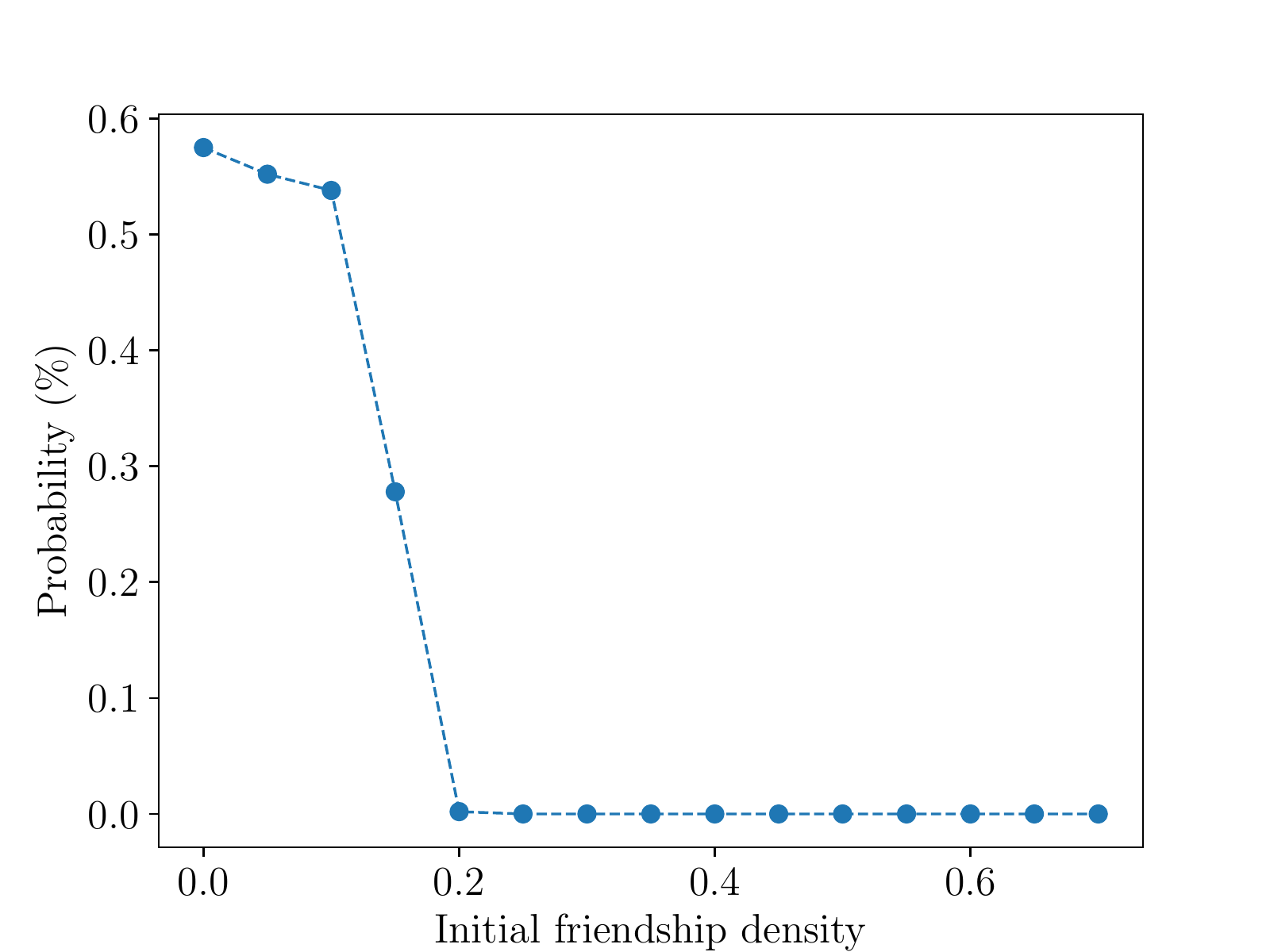}
    \caption{The jamming probability in CTD,
    for Barab\'asi-Albert networks with size $n=250$ and parameter $d\in[0,0.7]$
    exhibits a threshold behavior comparable to Erd\H{o}s-R\'{e}nyi graphs, but with significantly lower edge density.
    }\label{figure:BA_probability_density}
  \end{center}
\end{figure}

\begin{figure}
  \begin{center}
    \includegraphics[width=0.8\linewidth]{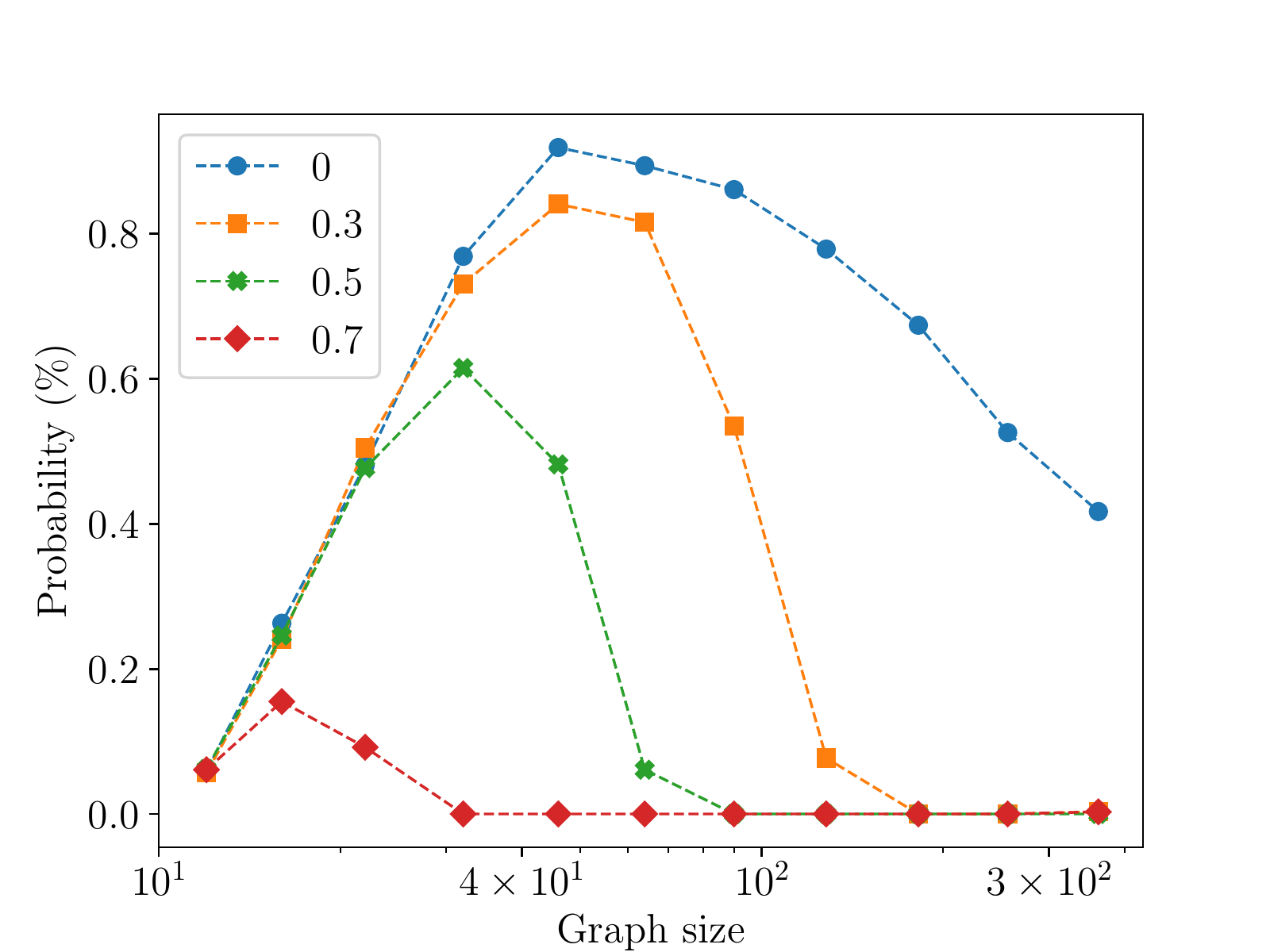}
    \caption{The jamming probability in CTD, for Barab\'asi-Albert networks with size $n\le 400$ and parameters $d\in\{0.0, 0.3, 0.5,0.7\}$.
    }\label{figure:BA_probability}
  \end{center}
\end{figure}

\section{Summary and Discussion}

The theory of structural balance provides a rigorous framework for the study of friendships and enmities in a population. 
A central concept in this theory has been the energy landscape of networks, and particularly the energy properties of its local minima.
In this paper we have taken a closer look at the properties of these local minima with respect to the stochastic process, addressing questions regarding their reachability and attractor properties.
We have shown that there are super-exponentially many jammed states, as opposed to the exponentially many balanced states,
and that any initial state that is not too friendship-dense can reach a jammed state.
Moreover, such jammed states are attractors, and hence cannot be escaped by random perturbations of the network.
These findings have strong implications for the socially-aware CTD process, which in fact gets stuck in such jammed states.

Motivated by these rich reachability and attractor properties of jammed states, we have introduced 
the plausible socially-aware dynamics BED.
We have shown that BED does not get stuck in jammed states and that it always reaches balance.
Moreover, we have seen that BED converges fast from many interesting states, such as those that are not too far from balance.

The new BED dynamics spawns some natural questions regarding its asymptotic behavior.
Although we have shown that BED converges fast (in $\calO(n^2)$ time) to balance from any state that is suitably close to balance, the general convergence rate remains open.

\new{An assumption made throughout our work is that the underlying network is complete.
That is, at each point in time, every two individuals have a defined relationship (they are either friends or enemies).
It is natural to consider non-complete underlying networks $U$, where only those pairs of individuals who are connected by an edge $e\in U$ have a defined relationship.
We note that this generalized setting is considerably more complicated.
First, one needs to adapt the notion of balance accordingly.
One way to do this is to say that a state is balanced if all cycles are balanced, where
a cycle in $U$ is balanced if it contains an even number of edges labeled ``$-$''.
While checking whether a current state is balanced can be done efficiently~\cite{harary1980simple},
several fundamental problems remain.
For instance, computing the distance to the closest balanced state is known to be intractable~\cite{Facchetti2011}.
As another example, to our knowledge the balanced states do not have any simple structure and even the complexity of computing their number (for a given non-complete underlying network $U$) is open.
As a final illustration, we note that there exist ``blinkers''~\cite{Antal2005}, that is,
states where CTD and BED get stuck repeating moves back and forth (rather than getting stuck being unable to make a move), see~\cref{figure:blinker}.
(Note that when the underlying network is complete, there are no blinkers for BED due to~\cref{thm:bedprob1}.)}
Investigating the properties of BED adapted to such generalized settings is thus left as an interesting direction for future research.

\begin{figure}
  \begin{center}
    \includegraphics[width=0.8\linewidth]{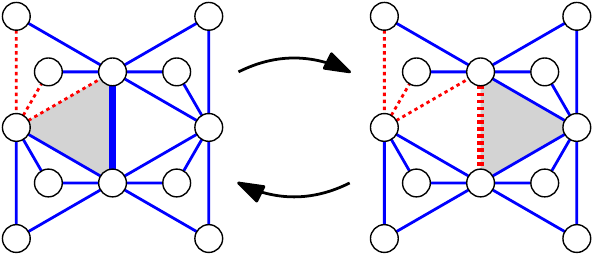}
    \caption{\new{A blinker. Under both CTD and BED, the thick edge in the middle keeps toggling between friendship (blue) and enmity (red) indefinitely. There is always only one imbalanced triangle (shaded) and flipping any other its edge would create more imbalanced triangles.}
    }\label{figure:blinker}
  \end{center}
\end{figure}

\section*{Acknowledgments}
K.C. acknowledges support from ERC Start grant no. (279307: Graph Games), ERC
Consolidator grant no. (863818: ForM-SMart), and Austrian Science Fund (FWF) grant no.
P23499-N23 and S11407-N23 (RiSE).
This project has received funding from the European Union’s Horizon 2020 research and innovation programme under the Marie Skłodowska-Curie Grant Agreement No. 665385.

All authors conceptualized the work.
A.P, J.S., and J.T. wrote the manuscript.
J.S. and J.T. produced the figures.
J.S. implemented the computer simulations.

\clearpage
\appendix

\section{Proofs for Reaching and Escaping the Jammed States}\label{app:ctd}

Here we formally state and prove our results on jammed states:
We present a family of new jammed states,
we show that the jammed states vastly outnumber the balanced ones,
and we establish the reachability properties stated in the main text.

\begin{definition}[Circular graph $S_k(n_0,\dots,n_{d-1})$]\label{def:circular}
Given an integer $k\ge0$ and a partition $n=n_0+\dots+n_{d-1}$ of $n$ into $d$ parts, the \emph{circular graph} $S_k(n_0,\dots,n_{d-1})$ is a signed graph consisting of $d$ clusters $V_0,\dots,V_{d-1}$ of sizes $n_0,\dots,n_{d-1}$, respectively, arranged along a circle in this order, such that the edge $(u,v)$ with $u\in V_i$, $v\in V_j$ is assigned a sign ``$+1$'' if and only if $V_i$ and $V_j$ are at most $k$ steps apart.
\end{definition}
In particular, when $n$ is divisible by 3 then the circular graph $S_0(n/3,n/3,n/3)$ corresponds to the jammed state~$J_n$ from~\cref{sec:ctd}.

\begin{theorem}[Counting jammed states]\label{thm:counting} There are at least $2^{\Omega(n\log n)}$ jammed states on $n$ labeled vertices.
\end{theorem}
\begin{proof}
The proof proceeds in two steps:
First, we show that for any signed graph $S_d(n_0,\dots,n_{4d+1})$ and an edge $(u,v)$ there are at least $2d+2$ clusters with the property that all the vertices from those clusters form a balanced triad with $(u,v)$.
In particular, this immediately implies that the signed graph $S^2_d$ with $n_0=\dots=n_{4d+1}=2$ is jammed:
Indeed, any edge $(u,v)$ is contained in at least $2(2d+2)-2=4d+2$ balanced triads (the $-2$ comes from omitting the vertices $u$, $v$ themselves) and in at most $2\cdot 2d=4d$ imbalanced triads.
Second, we show that there are $2^{\Omega(n\log n)}$ ways to draw the signed graph $S^2_d$ over the $n$ labelled vertices, hence at least $2^{\Omega(n\log n)}$ jammed states.

For the first part, suppose that the vertices $u$, $v$ belong to clusters that are $i$ steps apart.
We distinguish two cases.
\begin{enumerate}
\item $i\le d$ (that is, $(u,v)$ is labeled ``$+$''): Then there are $2d+1-i$ ``nearby'' clusters whose vertices $w$ form triads $(u,v,w)$ of type $\Delta_0$, and similarly $2d+1-i$  ``far-away'' clusters whose vertices $w$ form triads $(u,v,w)$ of type $\Delta_2$. In total, this is $4d+2-2i\ge 2d+2$ clusters with the desired property.
\item $i>d$ (that is, $(u,v)$ is labeled ``$-$''): Then there are $2i\ge 2d+2$ clusters ``nearby'' either $u$ or $v$ and ``far'' from the other vertex. All vertices $w$ from those clusters form triads $(u,v,w)$ of type $\Delta_2$.
\end{enumerate}

For the second part, we count only those jammed states in which each cluster has size 2.
Note that there are $n-1$ ways to pick a vertex to join the cluster of vertex $0$.
Then there are $\binom{n-2}2$ ways to select two vertices for the next (clockwise) cluster, then $\binom{n-4}2$ ways for the next cluster, and so on.
Finally, we must divide by 2, since the same signed graph would be obtained by selecting the vertices in the reverse order (or going counter-clockwise).
In total, using the Stirling approximation $n!\ge (n/e)^n$ and a trivial inequality $e\sqrt 2<4$,
we obtain that the number of different jammed states is at least
\begin{align*}
\quad\qquad \frac{(n-1)!}{2^{n/2}}&\ge \frac1n \frac{n!}{(\sqrt2)^n} \ge \frac1n (n/4)^n \\
 &= 2^{n\log_2 n - 2n -\log_2 n} = 2^{\Omega(n\log n)}. \qquad\quad \qedhere
\end{align*}
\end{proof}
Note that in comparison there are $2^{n-1}$ balanced states, since each balanced state is characterized by a subset of vertices of $\{1,\dots,n-1\}$ which are connected to vertex 0 by a friendship edge.
On the other hand, the total number of signed graphs is $2^{\binom n2}=2^{\Theta(n^2)}$.

Also, note that each of the $4d+2$ clusters of the jammed state $S^2_d$ constitutes a balanced clique, in the sense of~\cite{Marvel2009}.
This answers in affirmative an open question posed there:
For any $m\equiv 2\pmod 4$ there exists a jammed state with $m$ balanced cliques.

\bigskip
To prove the reachability properties, we define a specific jammed state $J_n$ on $n$ vertices labeled $1,\dots,n$:
The edges labeled ``$+$'' in $J$ form three roughly equal clusters:
One on vertices labeled $1,\dots,\fl{n/3}$, one on vertices labeled
$\fl{n/3}+1,\dots,2\fl{n/3}$, and one on the remaining vertices $2\fl{n/3}+1,\dots,n$.
It is easy to verify that for $n\ge 11$ the state $J_n$ is jammed.

\begin{theorem}[Reaching a jammed state]\label{thm:reaching} Let $G$ be a signed graph on $n \ge 11$ vertices such that each vertex is incident to at most $\frac{\n}{12}-1$ friendship edges. Then CTD reaches $J_n$ with positive probability.
\end{theorem}
Since $J_n$ is jammed, as a corollary we obtain that for any such $G$ the expected time to reach a balanced state is infinite.

\begin{proof}
For brevity, assume $n\equiv 0 \pmod 3$ (the other cases are completely analogous).
We describe a finite sequence of selected imbalanced triads and edge flips that results in $J_n$. Denote the three clusters of $J_n$ by $V_1$, $V_2$ and $V_3$, respectively.

First, we show that one by one, all the enmity edges within each cluster may be flipped into friendship edges.
Fix an enmity edge $e=(u,v)$ where $u,v\in V_i$.
It suffices to show that, throughout this phase, $e$ belongs to at most $n/2-2$ balanced triads.
Each balanced triad must contain a friendship edge incident to $u$ or to $v$.
Initially, there were $\frac{n}{12}-1$ friendship edges incident to $u$, and that many to $v$.
Moreover, throughout this phase, friendship edges to the other $n/3-2$ vertices within the cluster $V_i$ might have been added.
In total, this is at most $n/2-4$ vertices $w$ connected to one of $u$, $v$ by a friendship edge, thus at most $n/2-4$ balanced triads containing $e$. Hence $e$ can be flipped.

Second, we show that one by one, all friendship edges connecting vertices from different clusters may be flipped into enmities.
Suppose that $e=(u,v)$ is such a friendship edge.
It suffices to find $n/2$ imbalanced triads containing $(u,v)$.
Consider the other vertices in the cluster containing $u$.
They are all friends with $u$, but at most $n/12-1$ of them are friends with $v$
(since we never add friendship edges leading across clusters).
Thus there are at least $(n/3-1)-(n/12-1)=n/4$ imbalanced triads (of type $\Delta_1$)
containing $e$ and another vertex in the cluster of $v$.
Similarly, there are at least $\frac{n}{4}$ triads of type $\Delta_1$
defined by $e$ and another vertex in the cluster of $u$.
Hence $r_e\geq \frac{1}{2}n$, as claimed.

By flipping all friendship edges between different clusters to enmity edges, we reach a jammed state $J_n$ as claimed.
\end{proof}

\begin{theorem}[Escaping a jammed state]\label{thm:escaping}
Let $E_0$ be any set of edges such that each vertex is incident to at most $\n/12-1$ edges of $E_0$.
Let $S_n$ be a state obtained from $J_n$ by flipping the edges of $E_0$.
Then the CTD run from $S_n$ reaches $J_n$.
\end{theorem}
\begin{proof}
Again, without loss of generality, we assume that $n\equiv 0 \pmod 3$. 
We first show that no edge $e=(u,v)\not\in E_0$ can ever be flipped.
In $J_n$, any enmity edge belongs to $2n/3-2$ balanced triads
(and the friendship edges belong to even $n-2$ balanced triads).
Since $S_n$ differs from $J_n$ by at most $n/12-1$ edges incident to each vertex,
each edge $e\not\in E_0$ belongs to at least $2n/3-3 - 2(n/12-1)=n/2$ balanced triads in $S_n$
and thus can not be flipped.

On the other hand, any edge $e\in E_0$ belonged to at least $2n/3-2$ balanced triads in $J_n$,
thus it belongs to at least $2n/3-2 -2(n/12-1)=n/2$ imbalanced triads in $S_n$, and as such can be flipped.
Moreover, once such an edge has been flipped, by the above argument it cannot be flipped again. 
Hence CTD will flip each edge in $E_0$ once and return to the jammed state $J_n$.
\end{proof}

\section{Red-black graphs}\label{app:redblack}
The following lemma formalizes the key property of the red-black graphs.

\begin{lemma}\label{lemma:redblack}
	Let $G$ be a signed graph, let $C$ be a balanced state closest to $G$, let $R$ be the red-black graph associated to $G$ and $C$. Then a triad in $G$ is imbalanced if and only if exactly $1$ or $3$ of its edges are red in $R$.
\end{lemma}

\begin{proof}
To prove the lemma, we pick a triad in $R$ and check each of the $4$ possible cases:
\begin{itemize}
\item If a triad contains $0$ red edges, then all edge signs in $G$ agree with those in $C$ thus the triad is balanced in $G$ as $C$ is a balanced state.
\item If a triad contains $1$ red edge, we distinguish two cases. If both vertices of the red edge are in the same vertex class of $C$ (when treated as a bipartite graph w.r.t.~the friendship edges), then the triad is of type $\Delta_1$ in $G$. If vertices are in different vertex classes of $C$, then the triad is again of type $\Delta_1$ in $G$. Thus the triad is imbalanced in $G$.
\item If a triad contains $2$ red edges, we distinguish two cases. If all $3$ vertices of the triad are in the same vertex class of $C$, then the triad is of type $\Delta_2$ in $G$. If $2$ vertices are in one class and the third vertex is in the other, then depending on which two edges are red the triad is either of type $\Delta_0$ or $\Delta_2$ in $G$. Thus the triad is balanced in $G$.
\item If a triad contains $3$ red edges, we distinguish two cases. If all $3$ vertices of the triad are in the same vertex class of $C$, the triad is of type $\Delta_3$ in $G$. If $2$ vertices are in one class and the third vertex is in the other, the triad is of type $\Delta_1$ in $G$. Thus the triad is balanced in $G$. \qedhere
\end{itemize}
\end{proof}

\section{Fast convergence of BED from a jammed state}\label{app:fastbedstuck}

\begin{proposition}\label{prop:jammed-fast}
There exists a family of jammed states of increasing size $n$ such that BED starting in those states reaches a balanced state in $\calO(n^2)$ expected steps.	
\end{proposition}

\begin{proof}
	Let $n\ge 72$ be a positive integer divisible by $8$.
	Consider a circular graph $J'_n=S_1(x,y,y,x,y,y)$, where $x=|V_0| = |V_3| = \frac14n - 2$
	and $y=|V_i|=\frac18n+1$ for $i\in\{1,2,4,5\}$ (see~\cref{def:circular}).
	
First we show that $J'_n$ is jammed.
Consider a balanced state $B_n$ with parts $V_0\cup V_1\cup V_5$ and $V_2\cup V_3\cup V_4$ and
denote by $E_0$ the set of red edges in the corresponding red-black graph.
Note that there are no triads with all edges red, hence the rank of an edge is the number of triads that contain it and contain precisely one red edge.
For any red edge $(u,v)$ we have $r_{(u,v)}=2x$ due to $V_1$ and $V_4$.
Similarly, for any black edge $(u,v)$ such that $|\{u,v\}\cap (V_1\cup V_2\cup V_4\cup V_5)|=1$ we have $r_{(u,v)}=2y$
and for other black edges we have $r_{(u,v)}=0$.
Since all ranks are less than $\frac12n$, the state $J'_n$ is indeed jammed. 

Next, denote by $E_t$ the set of red edges (with respect to the same balanced state $B_n$) after $t$ steps of BED.
We will show that:
\begin{enumerate}
\item\label{itm:jammed-fast-1} $E_t\subseteq E_0$, and that
\item\label{itm:jammed-fast-2} at each point $t$ in time, $\Pr[|E_{t+1}|<|E_t|]\ge 2\cdot \Pr[|E_{t+1}|>|E_t|]$.
\end{enumerate}
Mapping the evolutionary dynamics to a one-dimensional random walk with a constant forward bias and an absorbing barrier corresponding to $|E_t|=0$, we thus conclude that the expected number of steps till balance is $\calO(|E_0|)=\calO(n^2)$.

To prove~\cref{itm:jammed-fast-1}, we proceed by induction.
Consider $E_t\subseteq E_0$. Note that, as before, there are no triads with all edges red. Also:
\begin{enumerate}
\item When $(u,v)\in E_t$ (that is, $(u,v)$ is red) then as before $r_{(u,v)}\ge 2x$ due to triads $(u,v,w)$ with $w\in V_0\cup V_3$.
\item When $(u,v)\not\in E_0 $ then $(u,v)$ is black and as before we have $r_{(u,v)}\le 2y$.
\end{enumerate}
Now consider any imbalanced triad. It contains a red edge. Since $2x>2y$, we always flip that red edge rather than any edge outside of $E_0$, thus $E_{t+1}\subseteq E_0$ as desired.
(Note that it is possible that we flip a black edge in $E_0$.)

To prove~\cref{itm:jammed-fast-2}, consider any time point $t$ and any red edge $(u,v)$.
We say that an imbalanced triad is \textit{good} if its red edge has a strictly higher rank than its other two edges, and \textit{bad} otherwise.
It suffices to show that $(u,v)$ belongs to twice as many good triads as bad triads.
Recall that for $w\in V_0\cup V_3$ the triad $(u,v,w)$ is good, hence $(u,v)$ belongs to at least $2x=2\cdot(\frac14n-2)\ge32$ good triads (here we use $n\ge 72$).

On the other hand, suppose that $(u,v,w)$ is a bad triad and without loss of generality, $(u,w)$ is the (black) edge with rank at least $2x$.
Note that $(u,w)$ belongs to $E_0\setminus E_t$ (other black edges have rank at most $2y$).
Denote by $d_i$ the \textit{red degree} of vertex $i$, that is, the number of red edges incident to $i$.
Then $2x\le r_{(u,w)}=d_u+d_w\le d_u+2y$, thus $d_u\ge 2x-2y$.
Vertex $u$ is therefore connected to at most $2y-(2x-2y)=4y-2x=8$ vertices in $E_0$ by a black edge.
Each such edge gives rise to at most one bad triad and likewise for the edge $(v,w)$, so in total $(u,v)$ belongs to at most
$2\cdot 8=16$ bad triads, concluding the proof.
\end{proof}

\section{Network descriptors for BED and CTD on Erd\H{o}s-R\'{e}nyi graphs}\label{app:table}
\new{Here we present the network descriptors when BED and CTD are run on Erd\H{o}s-R\'{e}nyi graphs with $n=128$ and $p\in\{0,0.4,0.5,0.6,0.7\}$, both before the process starts and after it finishes.}

\begin{center}
\begin{tabular}{c  c c c c}
Before\\
$p$ & $\overline{d}$ & $C$ & $\E[S]$ & $\Var[S]$\\
\hline
0&0&0&-&- \\
0.4&50.8&0.064&-&-\\
0.5&63.5&0.125&-&-\\
0.6&76.2&0.216&-&-\\
0.7&88.9&0.343&-&-\\
\hline
\end{tabular}
\end{center}

\begin{center}
\begin{tabular}{ c  c c c c}
After BED\\
$p$ & $\overline{d}$ & $C$ & $\E[S]$ & $\Var[S]$\\
\hline
0&63.160&0.246&61.468&3.820 \\
0.4&68.294&0.307&45.818&8.244\\
0.5&78.144&0.423&32.990&7.596\\
0.6&103.317&0.720&13.264&6.127\\
0.7&125.260&0.979&0.883&0.815\\
\hline
\end{tabular}
\end{center}

\smallskip

\begin{center}
\begin{tabular}{ c  c c c c}
After CTD\\
$p$ & $\overline{d}$ & $C$ & $\E[S]$ & $\Var[S]$\\
\hline
0&63.179&0.246&61.321&4.279\\
0.4&68.274&0.306&45.866&8.716\\
0.5&78.116&0.423&33.025&7.967\\
0.6&103.185&0.719&13.348&6.236\\
0.7&125.144&0.978&0.942&0.867\\
\hline
\end{tabular}
\end{center}
\new{As in the main text, we average over $10^5$ runs and exclude the runs of CTD that got jammed.
Here $\overline{d}$ is the average degree, $C$ is the clustering coefficient, $S$ is the size of the smaller clique once the process finishes, $\E[S]$ is its mean and $\Var[S]$ its variance.
The two dynamics match almost perfectly.}

\end{document}